\documentclass[12pt,a4paper]{amsart}

\usepackage{amsmath,amsfonts,amssymb,mathtools,extarrows}
\usepackage{xcolor}

\usepackage{color}
\usepackage{graphicx}

\usepackage{enumerate}  

\usepackage{cancel}

\usepackage[hmargin=2cm,vmargin=2cm]{geometry}

\usepackage[hypertexnames=false,hyperfootnotes=false,colorlinks=true,linkcolor=blue,%
citecolor=purple,filecolor=magenta,urlcolor=cyan,unicode,linktocpage=true,pagebackref=false]{hyperref}
\usepackage{nameref,zref-xr}     

\usepackage{tikz}

\newcommand{\mbZ}{\mathbb Z}
\newcommand{\mbC}{\mathbb C}

\newcommand{\oM}{\overline{\mathcal M}}
\newcommand{\tg}{\widetilde g}
\newcommand{\tu}{{\widetilde u}}

\def\oM{{\overline{\mathcal{M}}}}
\def\CP{{{\mathbb C}{\mathbb P}}}
\renewcommand{\Im}{\mathrm{Im}}

\def\d{{\partial}}

\newcommand{\eps}{\varepsilon}

\newcommand{\hcA}{\widehat{\mathcal A}}
\newcommand{\DR}{\mathrm{DR}}

\newcommand{\even}{\mathrm{even}}

\newcommand{\Coef}{\mathrm{Coef}}

\newcommand{\gl}{\mathrm{gl}}

\newcommand{\hu}{\widehat{u}}

\newcommand{\oP}{{\overline{P}}}

\newcommand{\triv}{\mathrm{triv}}
\newcommand{\KdV}{\mathrm{KdV}}

\newcommand{\bu}{\bar{u}}

\newcommand{\Id}{\mathrm{Id}}
\newcommand{\bfu}{\mathbf{u}}

\newcommand{\bfy}{\mathbf{y}}
\newcommand{\bfv}{\mathbf{v}}

\newtheorem{theorem}{Theorem}[section]
\newtheorem{proposition}[theorem]{Proposition}
\newtheorem{lemma}[theorem]{Lemma}

\theoremstyle{remark}

\newtheorem{remark}[theorem]{Remark}
\newtheorem{example}[theorem]{Example}

\theoremstyle{definition}

\newtheorem{definition}[theorem]{Definition}

\usepackage{color}

\renewcommand{\gg}[2]{\fill[color=white] (#2) circle(3mm) node {\color{black}$#1$}; \draw (#2) circle (3mm)}
\newcommand{\lab}[4]{\draw (#1)++(#2:#3) node {$#4$};}
\newcommand{\leg}[2]{\begin{scope}[shift={(#1)}] \draw (0:0) -- (#2:8mm);\end{scope}}

\numberwithin{equation}{section}

\begin{document}

\tikz{\coordinate (A) at (0,0);\coordinate (B) at (0,14mm);\coordinate (BL) at (-12mm,14mm);\coordinate (BR) at (12mm,14mm);}

\title[Reciprocal transformations and DR hierarchies]{An application of the nonlinear reciprocal transformations in the theory of DR hierarchies}

\author{Alexandr Buryak}
\address{A. Buryak:\newline 
Faculty of Mathematics, National Research University Higher School of Economics, \newline
6 Usacheva str., Moscow, 119048, Russian Federation;\smallskip\newline 
Center for Advanced Studies, Skolkovo Institute of Science and Technology, \newline
1 Nobel str., Moscow, 143026, Russian Federation}
\email{aburyak@hse.ru}

\author{Mikhail Troshkin}
\address{M. Troshkin:\newline 
Faculty of Mathematics, National Research University Higher School of Economics, \newline
6 Usacheva str., Moscow, 119048, Russian Federation}
\email{tr2mikh@gmail.com}

\begin{abstract}
We prove that the DR hierarchy corresponding to the family of F-cohomological field theories without unit considered in a previous work of the first author together with D. Gubarevich can be ``trivialized'', i.e. reduced to two copies of the KdV hierarchy, using a simple nonlinear reciprocal transformation. This gives the first manifestation of a role of nonlinear reciprocal transformation in the theory of integrable systems associated to the moduli spaces of stable curves, beyond the dispersionless limit.
\end{abstract}

\date{\today}

\maketitle

\section{Introduction}

There are various ways to produce integrable hierarchies of evolutionary PDEs using the geometry of the moduli space~$\oM_{g,n}$ of stable algebraic curves of genus $g$ with $n$ marked points. The central role here is played by the notion of a \emph{cohomological field theory (CohFT)} introduced by Kontsevich and Manin~\cite{KM94}. CohFTs are systems of cohomology classes on the moduli spaces $\oM_{g,n}$ that are compatible with natural maps between the moduli spaces. The notion of a CohFT involves also a vector space called the \emph{phase space}, a bilinear form on it called the \emph{metric}, and a special vector in the phase space called the \emph{unit}. 

\medskip

One way to produce an integrable hierarchy from a CohFT was proposed by Dubrovin and Zhang~\cite{DZ01} (for homogeneous semisimple CohFTs) and then generalized in~\cite{BPS12} (for general semisimple CohFTs). It is understood now~\cite{BS22} that a generalization of the Dubrovin--Zhang hierarchy exists for an object that is more general than a CohFT, for a so-called \emph{F-CohFT}, introduced in~\cite{BR21}, where the is no metric and there are less requirements regarding the compatibility with natural maps between the moduli spaces. All these more general hierarchies will be also called the \emph{Dubrovin--Zhang (DZ) hierarchies}. Note, however, that the polynomiality property of the DZ hierarchy is proved only for semisimple CohFTs~\cite{BPS12} and in some concrete examples~\cite{BR21,Bur23}.

\medskip

The DZ hierarchies include many important hierarchies from mathematical physics, for example, the Gelfand--Dickey hierarchies, Toda hierarchies of various types, the ILW hierarchy, the Drinfeld--Sokolov hierarchies, the discrete KdV hierarchy. Certain subclasses in the class of DZ hierarchies can be conjecturally described independently of the geometry, using only the language of integrable systems~\cite{DZ01,DLYZ16,LWZ21}.

\medskip

There is another way to produce an integrable hierarchy starting from a CohFT, which was proposed by the first author in~\cite{Bur15}, the resulting hierarchy was called the \emph{DR hierarchy}. The hierarchy is polynomial by construction, and it is endowed with a remarkably rich algebraic structure, which can be described very explicitly. It is understood now that the DR hierarchy can be associated to an arbitrary F-CohFT~\cite{BR21} (a systematic study is presented in~\cite{ABLR21}), and it is again polynomial by construction, while the polynomiality of the DZ hierarchy for an arbitrary F-CohFT is an open problem. Conjecturally, the DR and DZ hierarchies are Miura equivalent: for CohFTs it was formulated in~\cite{Bur15}, then in a more precise form in~\cite{BDGR18}, and currently the most general version of the conjecture is formulated in~\cite{BS22}.

\medskip

Note that the DR hierarchy can be associated to a more general object than an F-CohFT, to an \emph{F-CohFT without unit}: the first explicit examples were computed in~\cite{BG23}. A construction of a DZ hierarchy associated to an F-CohFT without unit is not developed yet.

\medskip

Having a construction of a class of integrable hierarchies, it is important to understand whether some of them are related by changes of variables or other transformations. A lot of work was done regarding the action of Miura transformations on the DZ and DR hierarchies. However, there is another type of transformations, the so-called \emph{reciprocal transformations}, whose role in the relation to the DZ and DR hierarchies is less studied. The dispersionless parts of these two hierarchies are hierarchies of hydrodynamic type, and the reciprocal transformations of such systems are well studied (see the references at the end of page 1 in~\cite{LSV23}). However, the reciprocal transformations of dispersive deformations of the hierarchies of hydrodynamic type are much less studied~(see e.g.~\cite{LZ11,LWZ23,LSV23}), and regarding the DZ or DR hierarchies there is only one work~\cite{LWZ23}, where the authors consider the simplest possible reciprocal transformations, the so-called \emph{linear} ones.

\medskip

In our paper, we give the first application of the \emph{nonlinear} reciprocal transformations in the theory of DR hierarchies. We consider the $3$-parameter family of F-CohFTs without unit of rank~$2$ from the paper~\cite{BG23} and the associated DR hierarchy. In~\cite{BG23}, the authors computed explicitly the primary flows of the DR hierarchy. In our paper, we give an explicit description of all the flows of the DR hierarchy: we prove that after the composition of a Miura transformation and a nonlinear reciprocal transformation the two dependent variables of the hierarchy become splitted and the resulting flows can be simply described in terms of the flows of the KdV hierarchy (see Theorem~\ref{theorem:main}). 

\medskip

Note that in the theory of DR and DZ hierarchies the KdV hierarchy is considered as the simplest possible hierarchy, because it corresponds to the trivial CohFT where all the classes are just units in the cohomology. As far as we know, the fact that the DR or DZ hierarchy corresponding to some CohFT or F-CohFT with nontrivial $R$-matrix can be ``trivialized'', i.e. reduced to the KdV hierarchy, was never observed in the literature before. So our result is the first one of this sort, and thus we believe that it shows the importance of the role of nonlinear reciprocal transformations in the theory of DR and DZ hierarchies, which should be clarified in the future research. 

\medskip

\subsection*{Notations and conventions}  

\begin{itemize}

\item We use the standard convention of sum over repeated Greek indices.

\smallskip

\item When it doesn't lead to a confusion, we use the symbol $*$ to indicate any value, in the appropriate range, of a sub- or superscript.

\smallskip

\item For a topological space $X$, we denote by $H^i(X)$ the cohomology groups with the coefficients in $\mbC$. Let $H^{\even}(X):=\bigoplus_{i\ge 0}H^{2i}(X)$.

\end{itemize}

\medskip

\subsection*{Acknowledgements}

The work of A.~B. is an output of a research project implemented as part of the Basic Research Program at the National Research University Higher School of Economics (HSE University). 

\medskip


\section{A family of F-CohFTs without unit of rank $2$ and the associated DR hierarchy}\label{section:a family of F-CohFTs}

Here we recall the construction of the DR hierarchy associated to an F-CohFT without unit and the main result from~\cite{BG23}.

\medskip

\begin{definition}
An \emph{F-cohomological field theory without unit} (F-CohFT without unit) is a system of linear maps 
$$
c_{g,n+1}\colon V^*\otimes V^{\otimes n} \to H^\even(\oM_{g,n+1}),\quad 2g-1+n>0,
$$
where $V$ is an arbitrary finite dimensional vector space, such that the following axioms are satisfied.
\begin{enumerate}[(i)]
\item The maps $c_{g,n+1}$ are equivariant with respect to the $S_n$-action permuting the $n$ copies of~$V$ in $V^*\otimes V^{\otimes n}$ and the last $n$ marked points on curves from $\oM_{g,n+1}$, respectively.

\smallskip

\item Fixing a basis $e_1,\ldots,e_{\dim V}$ in $V$ and the dual basis $e^1,\ldots,e^{\dim V}$ in $V^*$, the following property holds:
$$
\gl^* c_{g_1+g_2,n_1+n_2+1}(e^{\alpha_0}\otimes\otimes_{i=1}^{n_1+n_2} e_{\alpha_i}) = c_{g_1,n_1+2}(e^{\alpha_0}\otimes \otimes_{i\in I} e_{\alpha_i} \otimes e_\mu)\otimes c_{g_2,n_2+1}(e^{\mu}\otimes \otimes_{j\in J} e_{\alpha_j})
$$
for $1 \leq\alpha_0,\alpha_1,\ldots,\alpha_{n_1+n_2}\leq \dim V$, where $I \sqcup J = \{2,\ldots,n_1+n_2+1\}$, $|I|=n_1$, $|J|=n_2$, and $\gl\colon\oM_{g_1,n_1+2}\times\oM_{g_2,n_2+1}\to \oM_{g_1+g_2,n_1+n_2+1}$ is the corresponding gluing map. Clearly the axiom doesn't depend on the choice of a basis in $V$.
\end{enumerate}
The dimension of $V$ is called the \emph{rank} of the F-CohFT without unit.
\end{definition}

\medskip

We will use the following standard cohomology classes on $\oM_{g,n}$:
\begin{itemize}
\item The psi-class $\psi_i\in H^2(\oM_{g,n})$, $1\le i\le n$, is the first Chern class of the line bundle over~$\oM_{g,n}$ formed by the cotangent lines at the $i$-th marked point of stable curves.

\smallskip

\item The Hodge class $\lambda_j:= c_j(\mathbb E)\in H^{2j}(\oM_{g,n})$, $j\ge 0$, where~$\mathbb E$ is the rank~$g$ Hodge vector bundle over~$\oM_{g,n}$ whose fibers are the spaces of holomorphic one-forms on stable curves. 

\medskip

\item The \emph{double ramification (DR) cycle} $\DR_g(a_1,\ldots,a_n)\in H^{2g}(\oM_{g,n})$, $a_1,\ldots,a_n\in\mbZ$, $\sum a_i=0$, is defined as follows. There is a moduli space of projectivized stable maps to~$\CP^1$ relative to~$0$ and~$\infty$, with ramification profile over $0$ given by the negative numbers among the $a_i$-s, ramification profile over $\infty$ given by the positive numbers among the $a_i$-s, and the zeros among the $a_i$-s correspond to additional marked points (see, e.g.,~\cite{BSSZ15} for more details). This moduli space is endowed with a virtual fundamental class, which lies in homology of degree $2(2g-3+n)$. The DR cycle $\DR_g(a_1,\ldots,a_n)$ is the Poincar\'e dual to the pushforward, through the forgetful map to~$\oM_{g,n}$, of this virtual fundamental class. The crucial property of the DR cycle is that for any cohomology class $\theta\in H^*(\oM_{g,n})$ the integral $\int_{\oM_{g,n+1}}\lambda_g\DR_g\left(-\sum a_i,a_1,\ldots,a_n\right)\theta$ is a homogeneous polynomial in $a_1,\ldots,a_n$ of degree~$2g$ (see, e.g.,~\cite{Bur15}). 
\end{itemize}

\medskip

Let us briefly recall main notions and notations in the formal theory of evolutionary PDEs with one spatial variable:
\begin{itemize}
\item We fix an integer $N\ge 1$ and consider formal variables $u^1,\ldots,u^N$. To the formal variables $u^\alpha$ we attach formal variables $u^\alpha_d$ with $d\ge 0$ and introduce the algebra of \emph{differential polynomials} $\hcA_u:=\mbC[[u^*_0]][u^*_{\ge 1}][[\eps]]$. We identify $u^\alpha_0=u^\alpha$ and also denote $u^\alpha_x:=u^\alpha_1$, $u^{\alpha}_{xx}:=u^\alpha_2$, \ldots. Denote by $\hcA_{u;d}\subset\hcA_u$ the homogeneous component of (differential) degree $d$, where $\deg u^\alpha_i:=i$ and $\deg\eps:=-1$. 

\smallskip

\item An operator $\d_x\colon\hcA_u\to\hcA_u$ is defined by $\d_x:=\sum_{d\ge 0}u^\alpha_{d+1}\frac{\d}{\d u^\alpha_d}$.

\smallskip

\item An operator $H\colon\hcA_u\to\hcA_u$ is called \emph{evolutionary}, if $H$ satisfies the Leibniz rule and commutes with $\d_x$. An operator $H$ is evolutionary if and only if it has the form $H=H_{\oP}:=\sum_{n\ge 0}(\d_x^n P^\alpha)\frac{\d}{\d u^\alpha_n}$ for some $\oP=(P^1,\ldots,P^N)\in \hcA_u^N$. If the differential polynomials $P^\alpha$ satisfy the condition $P^\alpha|_{u^*_*=0}=0$, then we will write $H|_{u^*_*=0}=0$.

\smallskip

\item We assign to an evolutionary operator $H_{\oP}$ the system of \emph{evolutionary PDEs} (with one spatial variable) $\frac{\d u^\alpha}{\d t}=P^\alpha$, $1\le\alpha\le N$. Two such systems are said to be \emph{compatible} if the corresponding evolutionary operators commute.

\smallskip

\item An element $f\in\hcA_u$ is called a \emph{conservation law} for an evolutionary operator $H$ (or for the corresponding system of evolutionary PDEs) if there exists an element $R\in\hcA_u$ such that $H(f)=\d_x R$.

\smallskip

\item A \emph{Miura transformation} is a change of variables $u^\alpha\mapsto \tu^\alpha(u^*_*,\eps)$ of the form $\tu^\alpha(u^*_*,\eps)=g^\alpha(u^*_0)+\eps f^\alpha(u^*_*,\eps)$, where $f^\alpha\in\hcA_{u;1}$ and $g^\alpha\in\mbC[[u^*_0]]$ satisfy $g^\alpha|_{u^*_0=0}=0$ and $\left.\det\left(\frac{\d g^\alpha}{\d u^\beta_0}\right)\right|_{u^*_0=0}\ne 0$. 
\end{itemize}

\medskip

\begin{lemma}\label{lemma:uniqueness lemma}
Consider the case $N=1$ and denote $u_n:=u^1_n$. Consider a pair of compatible PDEs 
$$
\left\{
\begin{aligned}
&\frac{\d w}{\d t}=P,\\
&\frac{\d w}{\d s}=Q,
\end{aligned}
\right.
$$
where $P=uu_x+O(\eps)\in\hcA_{u;1}$, $Q=f(u)u_x+O(\eps)\in\hcA_{u;1}$, and $f\in\mbC[[u]]$. Then the differential polynomial $Q$ is uniquely determined by $P$ and $f$.
\end{lemma}
\begin{proof}
This is a slight generalization of~\cite[Lemma~4.14]{BR21}, with the same proof, so we omit~it.
\end{proof}

\medskip

Consider now an arbitrary F-CohFT without unit of rank $N$ and define differential polynomials $P^\alpha_{\beta,d}\in\hcA_u$, $1\le\alpha,\beta\le N$, $d\ge 0$, by
\begin{equation*}
P^\alpha_{\beta,d}:=\sum_{\substack{g,n\geq 0,\,2g+n>0\\k_1,\ldots,k_n\geq 0\\\sum_{j=1}^n k_j=2g}} \frac{\eps^{2g}}{n!} \Coef_{(a_1)^{k_1}\ldots(a_n)^{k_n}} \left(\int_{\DR_g(-\sum_{j=1}^n a_j,0,a_1,\ldots,a_n)} \hspace{-2.3cm}\lambda_g \psi_2^d c_{g,n+2}(e^\alpha\otimes e_\beta\otimes \otimes_{j=1}^n e_{\alpha_j}) \right)\prod_{j=1}^n u^{\alpha_j}_{k_j}.
\end{equation*}
The \emph{DR hierarchy} is the following system of evolutionary PDEs:
\begin{gather}\label{eq:DR hierarchy}
\frac{\d u^\alpha}{\d t^\beta_d}=\d_x P^\alpha_{\beta,d},\qquad 1\le \alpha,\beta\le N,\quad d\ge 0.
\end{gather}
All the equations of the DR hierarchy are pairwise compatible.

\medskip

\begin{example}\label{example:KdV}
Consider the trivial F-CohFT without unit given by $V=\mbC$, $e_1=1\in\mbC=V$, and
$$
c^\triv_{g,n+1}(e^1\otimes e_1^{\otimes n}):=1\in H^0(\oM_{g,n+1}).
$$
Then the corresponding DR hierarchy is the KdV hierarchy~\cite[Section~4.3.1]{Bur15} (we denote $u_d:=u^1_d$ and $t_d:=t^1_d$)
$$
\frac{\d u}{\d t_d}=\d_x P^\KdV_d,
$$
where 
$$
P^\KdV_0=u,\quad P^\KdV_1=\frac{u^2}{2}+\frac{\eps^2}{12}u_{xx},\quad P_2^\KdV=\frac{u^3}{6}+\frac{\eps^2}{24}(2u u_{xx}+u_x^2)+\frac{\eps^4}{240}u_{xxxx},
$$
and a general formula for $P^\KdV_d$ is
$$
\d_x P^\KdV_d=\frac{\eps^{2d+2}}{2(2d+1)!!}\left[\left(L^{d+\frac{1}{2}}\right)_+,L\right],\quad L=\d_x^2+2\eps^{-2}u,
$$
defining $\left.P^\KdV_d\right|_{u_*=0}:=0$. Note that
$$
P_d^{\KdV}=\frac{u^{d+1}}{(d+1)!}+O(\eps)\in\hcA_{u;0}.
$$
\end{example}

\medskip

In~\cite{BG23}, the authors considered the following family of F-CohFTs without unit, with phase space $V=\mbC^2$, parameterized by a vector $G=(G^1,G^2)\in\mbC^2$:
$$
c^{\triv,G}_{g,n+1}(e^{i_0}\otimes\otimes_{j=1}^n e_{i_j}):=
\begin{cases}
(G^{i_0})^g,&\text{if $i_0=i_1=\ldots=i_n$},\\
0,&\text{otherwise},
\end{cases}
$$
where $e_1,e_2$ is the standard basis of $\mbC^2$. Then the authors of~\cite{BG23} applied to this F-CohFT without unit the $R$-matrix $\Id+R_1z$ with 
\begin{gather*}
R_1=
\begin{pmatrix}
0 & \xi \\ 
0 & 0
\end{pmatrix},\quad\xi\in\mbC,
\end{gather*}
see the details about this action in~\cite{BG23}. The resulting F-CohFT without unit is denoted by $\left((\Id+R_1 z)c^{\triv,G}\right)_{g,n+1}$. In~\cite[Theorem~4.1]{BG23}, the authors considered the corresponding DR hierarchy and proved that after the Miura transformation 
\begin{gather*}
\tu^1=u^1+\xi\frac{(u^2)^2}{2}+\frac{\eps^2}{24}\d_x^2\left(\xi G^2 u^2+\frac{G^1}{1+\xi u^2}\right),\qquad \tu^2=u^2,
\end{gather*}
the flows $\frac{\d}{\d t^1_0}$ and $\frac{\d}{\d t^2_0}$ of the DR hierarchy become
\begin{align}
\frac{\d\tu^1}{\d t^1_0}=&\d_x\left[\frac{\tu^1}{1+\xi \tu^2}\right],\label{eq:primary flows1}\\
\frac{\d \tu^2}{\d t^1_0}=&0,\label{eq:primary flows2}\\
\frac{\d\tu^1}{\d t^2_0}=&\xi\d_x\left[\frac{\tu^1\tu^2}{1+\xi\tu^2}-\frac{1}{2}\frac{(\tu^1)^2}{(1+\xi\tu^2)^2}-\frac{\eps^2 G^1}{12}\left(\left(\left(\frac{\tu^1}{1+\xi\tu^2}\right)_x\frac{1}{1+\xi\tu^2}\right)_x\frac{1}{1+\xi\tu^2}\right)\right],\label{eq:primary flows3}\\
\frac{\d\tu^2}{\d t^2_0}=&\tu^2_x,\label{eq:primary flows4}
\end{align}
and moreover 
\begin{gather}\label{eq:equations for tu2}
\frac{\d\tu^2}{\d t^1_d}=0,\qquad \frac{\d\tu^2}{\d t^2_d}=\left.\d_x P_d^\KdV\right|_{u_n\mapsto\tu^2_n,\,\eps\mapsto\sqrt{G^2}\eps}.
\end{gather}

\medskip


\section{The reciprocal transformations and the main result}\label{section:reciprocal transformations and main result}

In this section we recall the definition of nonlinear reciprocal transformations and their main properties (see e.g.~\cite{LZ11}, however for completeness we present short proofs in the appendix).

\medskip

\begin{remark}
We will consider the nonlinear reciprocal transformations following the terminology of~\cite{LWZ23}. Since the reciprocal transformations of other types won't be considered, the adjective ``nonlinear'' will be omitted below.
\end{remark}

\medskip

Consider two $N$-tuples of variables $u^1,\ldots,u^N$ and $v^1,\ldots,v^N$, and the associated algebras of differential polynomials~$\hcA_u$ and $\hcA_v$. In the algebra $\hcA_v$, let us denote the spatial variable by~$y$, i.e. we denote $\d_y:=\sum_{n\ge 0}v^\alpha_{n+1}\frac{\d}{\d v^\alpha_n}$, and also $v^\alpha_y:=v^\alpha_1$, $v^\alpha_{yy}:=v^\alpha_2$, $\ldots$. Choose an element $f \in \hcA_{u;0}$ such that $f|_{u^*_*=0}=0$. Then the element $1+f\in\hcA_{u;0}$ is invertible. Define an algebra homomorphism $\Phi_f\colon\hcA_v\to\hcA_u$ by
\begin{gather*}
\Phi_f(P):=\left.P\right|_{v^\alpha_k\mapsto\left((1+f)^{-1}\d_x\right)^k(u^\alpha)},\quad P\in\hcA_v.
\end{gather*}
The homomorphism $\Phi_f$ is an isomorphism. It is called a \emph{reciprocal transformation}. We have
$$
(1+f)^{-1}\d_x\circ\Phi_f=\Phi_f\circ\d_y,
$$
and therefore under the isomorphism $\Phi_f\colon\hcA_v\to\hcA_u$ the operators $\d_y$ and $(1+f)^{-1}\d_x$ become identified.

\medskip

By abuse of notation, we will identify elements of $\hcA_v$ and their images under $\Phi_f$ in $\hcA_u$, as well as the operators $\d_y$ and $(1+f)^{-1}\d_x$. Note that under this identification evolutionary operators on $\hcA_v$ in general do not correspond to evolutionary operators on $\hcA_u$. However, suppose $f$ is a conservation law for an evolutionary operator $H$ on $\hcA_u$, so that $H(f) = \d_x R$ for some $R\in\hcA_u$. Then $H - R\d_y$ is an evolutionary operator on $\hcA_v$. Moreover, let us denote $H_1:=H$, $R_1:=R$, and suppose that $f$ is a conservation law for another evolutionary operator~$H_2$ on $\hcA_u$, $H_2(f) = \d_x R_2$, which commutes with $H_1$. Suppose also that $H_1|_{u^*_*=0}=H_2|_{u^*_*=0}=0$. Then the corresponding evolutionary operators $H_1 - R_1 \d_y$ and $H_2 - R_2 \d_y$ on $\hcA_v$ commute.

\medskip

We know that any evolutionary operator $H$ on $\hcA_u$ has the form $H=H_{\oP}$, where $\oP=(P^1,\ldots,P^N)\in\hcA_u^N$, and we assign to $H$ the system of PDEs 
\begin{gather}\label{eq:system of PDEs}
\frac{\d u^\alpha}{\d t}=P^\alpha,\quad 1\le\alpha\le N.
\end{gather}
Given a conservation law $f\in\hcA_{u;0}$ of $H$, $H(f)=\d_x R$, satisfying $f|_{u^*_*=0}=0$, we obtain the evolutionary operator $H-R\d_y$ on $\hcA_v$, to which we assign the system of PDEs 
\begin{gather}\label{eq:transformed system of PDEs}
\frac{\d v^\alpha}{\d t}=P^\alpha-R v^\alpha_y,\quad 1\le\alpha\le N.
\end{gather}
We will say that the system~\eqref{eq:transformed system of PDEs} is obtained from the system~\eqref{eq:system of PDEs} by the reciprocal transformation given by the conservation law $f$.

\medskip

\begin{remark}
Let us describe how the reciprocal transformations act on solutions of systems of PDEs. Consider a collection of pairwise commuting evolutionary operators $H_i$, $i\ge 1$, on~$\hcA_u$, $H_i=H_{\oP_i}$, and suppose that $\left.P^\alpha_i\right|_{u^*_*=0}=0$. Suppose that $f\in\hcA_{u;0}$ satisfying $f|_{u^*_*=0}=0$ is a common conservation law for the evolutionary operators $H_i$, $H_i f=\d_x R_i$. Consider a solution $(\bfu^1,\ldots,\bfu^N)\in\mbC[[x,t_*,\eps]]^N$, $\bfu^\alpha|_{x=t_*=0}=0$, of the system of PDEs
\begin{gather}\label{eq:system of PDEs-2}
\frac{\d u^\alpha}{\d t_i}=P_i^\alpha,\quad 1\le\alpha\le N,\, i\ge 1.
\end{gather}
Define a formal power series $\bfy\in\mbC[[x,t_*,\eps]]$ satisfying $\bfy|_{x=t_*=0}=0$ by the equation
$$
d\bfy=\left(1+f|_{u^\alpha_n=\d_x^n\bfu^\alpha}\right)dx+\sum_{i\ge 1}\left(R_i|_{u^\alpha_n=\d_x^n\bfu^\alpha}\right)dt_i.
$$
The $1$-form on the right-hand side is closed because $H_i(f)=\d_x f$ and $H_i(R_j)=H_j(R_i)$. The last equality is true because $\d_x(H_i(R_j) - H_j(R_i)) = H_i(H_j(f)) - H_j(H_i(f)) = 0$, which implies that $H_i(R_j) - H_j(R_i)\in\mbC[[\eps]]$. However, since $H_i|_{u^*_*=0}=H_j|_{u^*_*=0}=0$, we immediately obtain $H_i(R_j) - H_j(R_i)=0$. Let $\bfv^\alpha\in\mbC[[y,t_*,\eps]]$ be a unique formal power series satisfying 
$$
\bfv^\alpha|_{y=\bfy}=\bfu^\alpha.
$$ 
Then $(\bfv^1,\ldots,\bfv^N)$ is a solution of the system 
$$
\frac{\d v^\alpha}{\d t_i}=P^\alpha_i-R_i v^\alpha_y,\quad 1\le\alpha\le N,\,i\ge 1,
$$
which is obtained from the system~\eqref{eq:system of PDEs-2} by the reciprocal transformation given by the common conservation law $f$ of the evolutionary operators $H_i$.
\end{remark}

\medskip

\begin{example}
Consider the KdV hierarchy
$$
\frac{\d u}{\d t_d}=\d_x P_d^{\KdV},\quad d\ge 0.
$$
We see that $u$ is a common conservation law of the flows of the KdV hierarchy. So for any $\xi\in\mbC$ we have the reciprocal transformation of the KdV hierarchy given by the conservation law $\xi u$:
$$
\frac{\d v}{\d t_d}=\underbrace{\d_x P_d^{\KdV}-\xi P_d^{\KdV}v_y}_{=:Q_d^{\xi-\KdV}\in\hcA_v},\quad d\ge 0.
$$
Since $u^2$ is a conservation law of the KdV hierarchy, we obtain $u\d_x P_d^{\KdV}\in\Im(\d_x)$, which implies that $(1+\xi u)\d_x P_d^{\KdV}-\xi P_d^{\KdV}u_x\in\Im(\d_x)$, and therefore $Q_d^{\xi-\KdV}\in\Im(\d_y)$. Thus, the reciprocal transformation of the KdV hierarchy given by the conservation law $\xi u$ has the form 
$$
\frac{\d v}{\d t_d}=\d_y P_d^{\xi-\KdV},\quad d\ge 0,
$$ 
where $\left.P_d^{\xi-\KdV}\right|_{v_*=0}=0$. For example, 
$$
P_0^{\xi-\KdV}=v,\qquad P_1^{\xi-\KdV}=\frac{v^2}{2}+\xi\frac{v^3}{6}+\frac{\eps^2}{12}(1+\xi v)^3v_{yy}.
$$
\end{example}

\medskip

Consider again the DR hierarchy associated to the F-CohFT without unit $\left((\Id+R_1 z)c^{\triv,G}\right)_{g,n+1}$ from Section~\ref{section:a family of F-CohFTs}. Since $P^2_{1,d}=0$ and $P^2_{2,d}$ doesn't depend on $u^1_n$ for any $n\ge 0$, after any Miura transformation that doesn't change the variable $u^2$, this variable will be a conservation law for the transformed hierarchy.

\medskip

\begin{theorem}\label{theorem:main}
Consider the F-CohFT without unit $\left((\Id+R_1 z)c^{\triv,G}\right)_{g,n+1}$ and the associated DR hierarchy. Then the composition of the Miura transformation
\begin{gather}\label{eq:main Miura transformation}
\hu^1=\frac{1}{1+\xi u^2}\left(u^1+\xi\frac{(u^2)^2}{2}+\frac{\eps^2}{24}\d_x^2\left(\xi G^2 u^2+\frac{G^1}{1+\xi u^2}\right)\right),\qquad \hu^2=u^2,
\end{gather}
and the reciprocal transformation given by the conservation law $\xi\hu^2$ transforms the DR hierarchy to the system
\begin{align*}
&\frac{\d v^1}{\d t^1_d}=\left.\left(\d_x P_d^{\KdV}\right)\right|_{u_l\mapsto v^1_l,\,\eps\mapsto\sqrt{G^1}\eps}, && \frac{\d v^2}{\d t^1_d}=0,\\
&\frac{\d v^1}{\d t^2_d}=-\left.\left(\xi\d_x P_{d+1}^{\KdV}\right)\right|_{u_l\mapsto v^1_l,\,\eps\mapsto\sqrt{G^1}\eps}, && \frac{\d v^2}{\d t^2_d}=\left.\left(\d_y P_d^{\xi-\KdV}\right)\right|_{v_l\mapsto v^2_l,\,\eps\mapsto\sqrt{G^2}\eps}.
\end{align*}
\end{theorem}

\medskip


\section{Proof of Theorem~\ref{theorem:main}}

\subsection{Computation of the dispersionless part of the hierarchy}

Let us recall how to compute the dispersionless part, i.e. when $\eps=0$, of the DR hierarchy associated to an arbitrary F-CohFT without unit $\{c_{g,n+1}\colon V^*\otimes V^{\otimes n}\to H^\even(\oM_{g,n+1})\}$. For this, define an $N$-tuple of formal power series $(F^1,\ldots,F^N)\in\mbC[[u^1,\ldots,u^N]]^N$ by
$$
F^\alpha:=\sum_{n\ge 2}\frac{1}{n!}\sum_{1\le\alpha_1,\ldots,\alpha_n\le N}\left(\int_{\oM_{0,n+1}}c_{0,n+1}(e^\alpha\otimes \otimes_{j=1}^n e_{\alpha_j})\right)\prod_{j=1}^n u^{\alpha_j}.
$$
Let $c^\alpha_{\beta\gamma} = \frac{\d^2 F^\alpha}{\d u^\beta \d u^\gamma}$ and define $N\times N$ matrices $C_{\gamma}:=(c_{\beta\gamma}^\alpha)_{1\le\alpha,\beta\le N}$. Then the $N\times N$ matrices $P^{[0]}_d:=(P^\alpha_{\beta,d}|_{\eps=0})_{1\le\alpha,\beta\le N}$ are uniquely determined by the relations
\begin{gather}\label{eq:relations for genus 0}
\frac{\d P^{[0]}_{d}}{\d u^\gamma} = C_\gamma \cdot P^{[0]}_{d-1}, \quad \left.P^{[0]}_d\right|_{u^*=0}=0,\quad d\ge 0,\,\, 1\le\gamma\le N,
\end{gather}
where $P_{-1}^{[0]}:=\Id$.

\medskip

Let us return to our F-CohFT without unit $\left((\Id+R_1 z)c^{\triv,G}\right)_{g,n+1}$.
\begin{lemma}
We have
$$
F^1(u^1,u^2) = \frac{\frac{(u^1)^2}{2} +  u^1 \cdot \frac{\xi (u^2)^2}{2} - \frac{\xi (u^2)^3}{24}(4+\xi u^2)}{1+\xi u^2}, \qquad F^2(u^1,u^2) = \frac{(u^2)^2}{2}.
$$
\end{lemma}
\begin{proof}
We proceed in the same way as in~\cite[proof of Theorem~4.1]{BG23}: in order to compute the integrals
\begin{gather*}
\int_{\oM_{0,n+1}}\left((\Id+R_1 z)c^{\triv,G}\right)_{0,n+1}(e^\alpha\otimes \otimes_{j=1}^n e_{\alpha_j}),
\end{gather*}
we express the class $\left((\Id+R_1 z)c^{\triv,G}\right)_{0,n+1}(e^\alpha\otimes \otimes_{j=1}^n e_{\alpha_j})$ as a sum over stable trees. All the trees that give a nontrivial contribution to the integral are depicted in Figure~\ref{figure:graphs}, where the first four trees contribute to $F^1(u^1,u^2)$, and the last one contributes to $F^2(u^1,u^2)$. The respective sums are
\begin{equation*}
F^1(u^1,u^2) = \frac{(u^1)^2/2}{1+\xi u^2} + \frac{u^1 \cdot \xi (u^2)^2/2}{1+\xi u^2} + \frac{\xi^2 (u^2)^4/8}{1+\xi u^2} - \frac{\xi (u^2)^3}{6}, \qquad F^2(u^1,u^2)=\frac{(u^2)^2}{2}.
\end{equation*}
\begin{figure}[t]
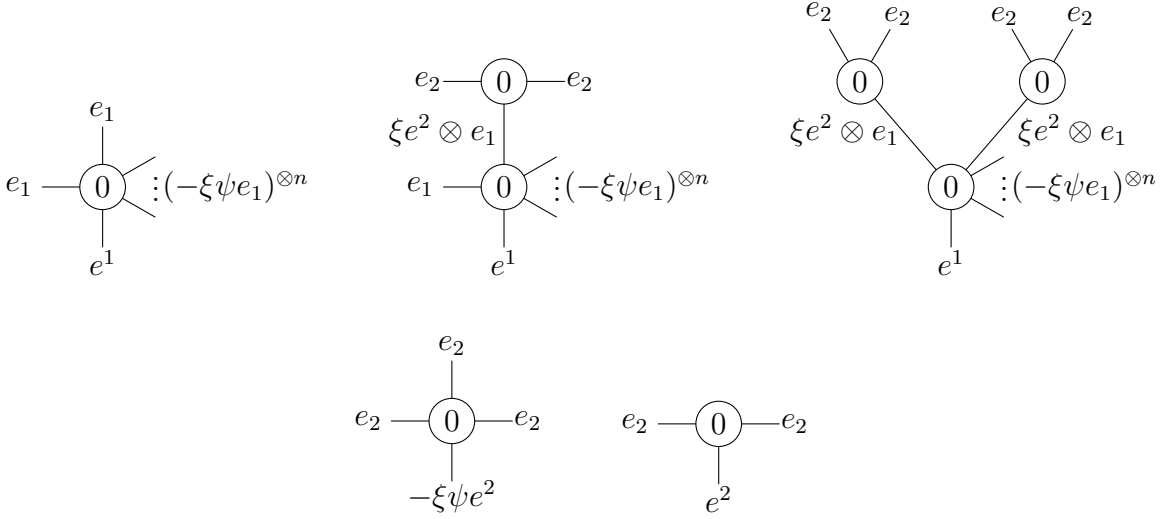

\tikz{
    \leg{A}{30};
    \leg{A}{-30};
    \leg{A}{-90};
    \leg{A}{180};
    \leg{A}{90};
    \gg{0}{A};
    \lab{A}{0}{7mm}{\raisebox{0.3\height}{\vdots}};
    \lab{A}{0}{17.5mm}{\raisebox{0.3\height}{$(-\xi \psi e_1)^{\otimes n}$}};
    \lab{A}{-90}{10mm}{e^1};
    \lab{A}{90}{10mm}{e_1};
    \lab{A}{180}{11mm}{e_1};
}
\hspace{0.5cm}
\tikz{
    \draw (A)--(B);
    \leg{A}{30};
    \leg{A}{-30};
    \leg{A}{-90};
    \leg{A}{180};
    \leg{B}{0};
    \leg{B}{180};
    \gg{0}{A};
    \gg{0}{B};
    \lab{A}{0}{7mm}{\raisebox{0.3\height}{\vdots}};
    \lab{A}{0}{17.5mm}{\raisebox{0.3\height}{$(-\xi \psi e_1)^{\otimes n}$}};
    \lab{A}{-90}{10mm}{e^1};
    \lab{B}{0}{10mm}{e_2};
    \lab{B}{180}{10mm}{e_2};
    \lab{A}{180}{11mm}{e_1};
    \node at (-8mm,7.5mm) {$\xi e^2\otimes e_1$};
}
\hspace{0.5cm}
\tikz{
    \draw (A)--(BL);
    \draw (A)--(BR);
    \leg{A}{30};
    \leg{A}{-30};
    \leg{A}{-90};
    \leg{BL}{120};
    \leg{BL}{60};
    \leg{BR}{120};
    \leg{BR}{60};
    \gg{0}{A};
    \gg{0}{BL};
    \gg{0}{BR};
    \lab{A}{0}{7mm}{\raisebox{0.3\height}{\vdots}};
    \lab{A}{0}{17.5mm}{\raisebox{0.3\height}{$(-\xi \psi e_1)^{\otimes n}$}};
    \lab{A}{-90}{10mm}{e^1};
    \lab{BL}{60}{10mm}{e_2};
    \lab{BL}{120}{10.7mm}{e_2};
    \lab{BR}{60}{10mm}{e_2};
    \lab{BR}{120}{10mm}{e_2};
    \node at (-14mm,7.5mm) {$\xi e^2\otimes e_1$};
    \node at (16mm,7.5mm) {$\xi e^2\otimes e_1$};
}
\\[0.5cm]
\tikz{
    \leg{A}{0};
    \leg{A}{-90};
    \leg{A}{180};
    \leg{A}{90};
    \gg{0}{A};
    \lab{A}{0}{10mm}{e_2};
    \lab{A}{-90}{10mm}{-\xi\psi e^2};
    \lab{A}{90}{10mm}{e_2};
    \lab{A}{180}{11mm}{e_2};
}
\hspace{0.5cm}
\tikz{
    \leg{A}{-90};
    \leg{A}{180};
    \leg{A}{0};
    \gg{0}{A};
    \lab{A}{-90}{10mm}{e^2};
    \lab{A}{0}{10mm}{e_2};
    \lab{A}{180}{11mm}{e_2};
}
\caption{Stable trees contributing to $(F^1,F^2)$}
\label{figure:graphs}
\end{figure}
\end{proof}

\begin{lemma}
\label{lemma:dispersionless part}
For any $d\ge 0$, we have
\begin{equation*}
P^{[0]}_d = \begin{pmatrix}\frac{(\bu^1)^{d+1}}{(d+1)!} && M_d \\ 0 && \frac{(\bu^2)^{d+1}}{(d+1)!}\end{pmatrix},
\end{equation*}
where
\begin{equation*}
\bu^1 := \hu^1|_{\eps=0} = \frac{u^1 + \frac{\xi (u^2)^2}{2}}{1 + \xi u^2},\hspace{0.2cm} \bu^2:=u^2,\hspace{0.2cm} M_d:=-\frac{\xi \left((\bu^1)^{d+2} - (d+2)\bu^1 (\bu^2)^{d+1} + (d+1)(\bu^2)^{d+2}\right)}{(d+2)!}.
\end{equation*}
\end{lemma}
\begin{proof}
We compute
\begin{equation}
C_1 = \begin{pmatrix}\frac{1}{1+\xi \bu^2} && \frac{\xi (\bu^2 - \bu^1)}{1+\xi \bu^2}\\ 0 && 0\end{pmatrix}, \qquad C_2 = \begin{pmatrix}  \frac{\xi (\bu^2 - \bu^1)}{1+\xi \bu^2} && \frac{\xi (\bu^1-\bu^2)(1 + \xi \bu^1)}{1+\xi \bu^2}\\ 0 && 1\end{pmatrix}.
\end{equation}
Then one directly checks the relations~\eqref{eq:relations for genus 0}, expressing 
\begin{equation*}
\frac{\d}{\d u^1}=\frac{1}{1+\xi \bu^2}\frac{\d}{\d \bu^1}, \qquad \frac{\d}{\d u^2}=\xi\frac{\bu^2-\bu^1}{1+\xi \bu^2}\frac{\d}{\d \bu^1}+\frac{\d}{\d \bu^2}.
\end{equation*}
\end{proof}

\medskip

\begin{proposition}\label{proposition:dispersionless part}
After the composition of the Miura transformation~\eqref{eq:main Miura transformation} and the reciprocal transformation given by the conservation law $\xi\hu^2$, the DR hierarchy of Theorem~\ref{theorem:main} has the form
\begin{align*}
&\frac{\d v^1}{\d t^1_d} = \d_y \left( \frac{(v^1)^{d+1}}{(d+1)!}\right)+O(\eps), && \frac{\d v^2}{\d t^1_d} = 0,\\
&\frac{\d v^1}{\d t^2_d} = -\xi \d_y \left(\frac{(v^1)^{d+2}}{(d+2)!}\right)+O(\eps), && \frac{\d v^2}{\d t^2_d} = \left.\left(\d_y P_d^{\xi-\KdV}\right)\right|_{v_l\mapsto v^2_l,\,\eps\mapsto\sqrt{G^2}\eps}.
\end{align*}
\end{proposition}
\begin{proof}
This is a direct computation. We apply the Miura transformation~\eqref{eq:main Miura transformation} to the dispersionless part of the DR hierarchy, computed in Lemma~\ref{lemma:dispersionless part}, and obtain
\begin{align*}
&\frac{\d\hu^1}{\d t^1_d}=\frac{1}{1+\xi\hu^2}\d_x\left(\frac{(\hu^1)^{d+1}}{(d+1)!}\right)+O(\eps), && \frac{\d\hu^2}{\d t^1_d} = 0,\\
&\frac{\d\hu^1}{\d t^2_d} = -\xi\frac{\hu^1_x \left((\hu^1)^{d+1} - (\hu^2)^{d+1}\right)}{(d+1)! (1+\xi \hu^2)}+O(\eps), && \frac{\d\hu^2}{\d t^2_d} = \left.\d_x P_d^\KdV\right|_{u_n\mapsto\hu^2_n,\,\eps\mapsto\sqrt{G^2}\eps}.
\end{align*}
One can easily see that after the reciprocal transformation given by the conservation law $\xi\hu^2$ the hierarchy has the desired form.
\end{proof}

\medskip

\subsection{The full hierarchy}

We are ready to finish the proof of Theorem~\ref{theorem:main}.

\medskip

If $\xi=0$, then our F-CohFT without unit coincides with $c^{\triv,G}_{g,n+1}$, and the theorem follows from Example~\ref{example:KdV}.

\medskip

Suppose that $\xi\ne 0$. We know that after the composition of the Miura transformation~\eqref{eq:main Miura transformation} and the reciprocal transformation given by the conservation law $\xi\hu^2$, the DR hierarchy has the form
\begin{align*}
&\frac{\d v^1}{\d t^1_d} = S_d, && \frac{\d v^2}{\d t^1_d} = 0,\\
&\frac{\d v^1}{\d t^2_d} = T_d, && \frac{\d v^2}{\d t^2_d} = \left.\left(\d_y P_d^{\xi-\KdV}\right)\right|_{v_l\mapsto v^2_l,\,\eps\mapsto\sqrt{G^2}\eps},
\end{align*}
where 
$$
S_d=\d_y \left( \frac{(v^1)^{d+1}}{(d+1)!}\right)+O(\eps)\in\hcA_{v;1},\qquad T_d=-\xi \d_y \left(\frac{(v^1)^{d+2}}{(d+2)!}\right)+O(\eps)\in\hcA_{v;1}.
$$
A direct computation using formulas~\eqref{eq:primary flows1}--\eqref{eq:primary flows4} gives that
$$
S_0=v^1_y,\qquad T_0=-\xi\d_y\left(\frac{(v^1)^2}{2}+\frac{\eps^2G^1}{12}v^1_{yy}\right).
$$
Let us prove that $S_d$ and $T_d$ don't depend on the variables~$v^2_n$. Indeed, the fact that the flows~$\frac{\d}{\d t^1_0}$ and~$\frac{\d}{\d t^1_d}$ commute gives that
$$
0=\frac{\d}{\d t^1_d}\frac{\d v^1}{\d t^1_0}-\frac{\d}{\d t^1_0}\frac{\d v^1}{\d t^1_d}=\d_y S_d-\sum_{n\ge 0}v^1_{n+1}\frac{\d S_d}{\d v^1_n}=\sum_{n\ge 0}v^2_{n+1}\frac{\d S_d}{\d v^2_n},
$$
which implies that $\frac{\d S_d}{\d v^2_n}=0$ for any $n$. In the same way, the commutativity of the flows~$\frac{\d}{\d t^1_0}$ and~$\frac{\d}{\d t^2_d}$ gives that $\frac{\d T_d}{\d v^2_n}=0$ for any $n$. Since the flows of the KdV hierarchy pairwise commute, Proposition~\ref{proposition:dispersionless part} and Lemma~\ref{lemma:uniqueness lemma} imply that $S_d=\left.\left(\d_x P_d^{\KdV}\right)\right|_{u_l\mapsto v^1_l,\,\eps\mapsto\sqrt{G^1}\eps}$ and $T_d=-\left.\left(\xi\d_x P_{d+1}^{\KdV}\right)\right|_{u_l\mapsto v^1_l,\,\eps\mapsto\sqrt{G^1}\eps}$. This completes the proof of Theorem~\ref{theorem:main}.

\medskip

{\appendix

\section{Reciprocal transformations}

Here, for completeness, we present short proofs of the properties of reciprocal transformations mentioned in Section~\ref{section:reciprocal transformations and main result}, see Parts~1 and~2 of Proposition~\ref{proposition:reciprocal transformations}. There are also Part 3 in Proposition~\ref{proposition:reciprocal transformations} and Remark~\ref{remark:action by conservation laws}, which we think are of independent interest, and which as far as we know didn't appear in the literature before.

\medskip 

\begin{proposition}\label{proposition:reciprocal transformations}
Suppose $f \in \hcA_{u;0}$ satisfying $f|_{u^*_*=0}=0$ is a conservation law for an evolutionary operator $H$ on~$\hcA_u$, so that $H(f) = \d_x R$ for some $R\in\hcA_u$. Consider the associated reciprocal transformation $\Phi_f\colon\hcA_v\to\hcA_u$.
\begin{enumerate}
\item[1.] $H - R\d_y$ is an evolutionary operator on~$\hcA_v$.

\smallskip

\item[2.] Denote $H_1:=H$, $R_1:=R$, and suppose that $f$ is a conservation law for another evolutionary operator $H_2$ on $\hcA_u$, $H_2(f) = \d_x R_2$, which commutes with $H_1$. Suppose also that $H_1|_{u^*_*=0}=H_2|_{u^*_*=0}=0$. Then the evolutionary operators $H_1 - R_1 \d_y$ and $H_2 - R_2 \d_y$ on $\hcA_v$ commute.

\smallskip

\item[3.] The map $g\mapsto \frac{g}{1+f}$ gives a one-to-one correspondence between the conservation laws of the operator $H$ on $\hcA_u$ and the conservation laws of the operator $H-R\d_y$ on $\hcA_v$. Moreover, if $H(g)=\d_x R_g$, then $(H - R \d_y)\big(\frac{g}{1+f}\big)=\d_y\big(R_g - \frac{g R}{1+f}\big)$.
\end{enumerate}
\end{proposition}
\begin{proof}
\emph{1}. The Leibniz rule holds for $H - R\d_y$ by linearity, so we only need to check the equality $[H - R\d_y, \d_y] = 0$. Indeed,
\begin{multline*}
[H - R\d_y, \d_y]= \left[H, \d_y \right] - [R\d_y, \d_y]= \left[H, \frac{1}{1+f}\d_x \right] - [R\d_y, \d_y]=\\
= H\left(\frac{1}{1+f}\right) \d_x + (\d_y R) \d_y=-\frac{H(f)}{(1+f)^2} \d_x + \frac{\d_x R}{(1+f)^2}  \d_x = 0.
\end{multline*}

\medskip

\emph{2}. We observe that $H_2(R_1) - H_1(R_2) = 0$. Indeed, $\d_x(H_2(R_1) - H_1(R_2)) = H_2H_1(f) - H_1 H_2(f) = 0$, which implies that $H_2(R_1) - H_1(R_2)\in\mbC[[\eps]]$. However, since $H_1|_{u^*_*=0}=H_2|_{u^*_*=0}=0$, we immediately obtain $H_2(R_1) - H_1(R_2)=0$. We further compute
\begin{align*}
&[H_1 - R_1\d_y, H_2 - R_2\d_y] = [H_1, H_2] - [H_1, R_2\d_y] - [R_1\d_y, H_2] + [R_1 \d_y, R_2\d_y]= \\
     &\hspace{4.23cm}= -\left[H_1, \frac{R_2}{1+f}\d_x\right] - \left[\frac{R_1}{1+f}\d_x, H_2\right] + \left[R_1\d_y, R_2\d_y\right],\\
&-\left[H_1, \frac{R_2}{1+f}\d_x\right] = \left( -\frac{H_1(R_2)}{1+f} + \frac{R_2 \cdot H_1(f)}{(1+f)^2}\right)\d_x = \left( -\frac{H_1(R_2)}{1+f} + \frac{R_2 \cdot \d_x R_1}{(1+f)^2}\right)\d_x,\\
&- \left[\frac{R_1}{1+f}\d_x, H_2\right] = \left(\frac{H_2(R_1)}{1+f} - \frac{R_1\cdot H_2(f)}{(1+f)^2}\right)\d_x = \left(\frac{H_2(R_1)}{1+f} - \frac{R_1\cdot \d_x R_2}{(1+f)^2}\right)\d_x,\\
&\left[R_1\d_y, R_2\d_y\right] = \left(R_1\cdot \d_y R_2 - R_2\cdot\d_y R_1 \right)\d_y = \frac{1}{(1+f)^2}\left(R_1 \cdot \d_x R_2 - R_2\cdot \d_x R_1 \right)\d_x.
 \end{align*}
We see that all terms in the resulting sum cancel out.

\medskip

\emph{3}. If $H(g)=\d_x R_g$, then we have
\begin{align*}
(H - R \d_y)\left(\frac{g}{1+f}\right)&= \frac{H(g)}{1+f} - \frac{g \cdot H(f)}{(1+f)^2} - \frac{R \cdot \d_y g}{1+f} + \frac{gR \cdot \d_y f}{(1+f)^2}=\\
&= \d_y R_g - \frac{g \cdot \d_y R}{1+f}-\frac{R \cdot \d_y g}{1+f} + \frac{gR \cdot\d_y f}{(1+f)^2}=\\
&= \d_y \left(R_g - \frac{g R}{1+f}\right).
\end{align*}
Conversely, if $(H-R\d_y)\tg=\d_y R_{\tg}$, then in the same way we check that $H\left((1+f)\tg\right)=\d_x\left(R_{\tg}+\tg R\right)$.
\end{proof}

\medskip

\begin{remark}\label{remark:action by conservation laws}
We see that the reciprocal transformations give an action of the set $\mathrm{CL}(H):=\{f\in\hcA_{u;0}|H(f)\in\Im(\d_x),\,f|_{u^*_*=0}=0\}$ on the evolutionary operator $H$: given $f\in\mathrm{CL}(H)$, we transform the operator $H$ to the operator $H-R_f\d_y$, where $R_f$ is given by $H(f)=\d_x R_f$ with $R_f|_{u^*_*=0}=0$. Since, by Part 3 of the proposition, there is a one-to-one correspondence between the conservation laws of the operator $H$ and the conservation laws of the operator $H-R_f\d_y$, we can further act on $H-R_f\d_y$ by any $g\in\mathrm{CL}(H)$, meaning that we act on $H-R_f\d_y$ by $\frac{g}{1+f}\in\mathrm{CL}(H-R_f\d_y)$. Let us compute the composition of two such transformations. Indeed, for this we consider another $N$-tuple of formal variables $w^1,\ldots,w^N$ and the reciprocal transformation $\Phi_{\frac{g}{1+f}}\colon\hcA_w\to\hcA_v$. Let us denote the spatial variable in the algebra $\hcA_w$ by $z$. Acting by the conservation law $\frac{g}{1+f}\in\mathrm{CL}(H-R_f\d_y)$ on the operator $H-R_f\d_y$, we obtain the evolutionary operator $H-R_f\d_y-\left(R_g-\frac{g R_f}{1+f}\right)\d_z$ on $\hcA_w$, and a direct computation shows that it is equal to the operator $H-(R_f+R_g)\d_z$. So the result is the same as the result of the action on $H$ by $f+g\in\mathrm{CL}(H)$. Thus, the constructed action of the set $\mathrm{CL}(H)$ on the evolutionary operator $H$ is actually a group action with the respect to the standard additive group structure on $\mathrm{CL}(H)$. 
\end{remark}
}


\end{document}